\documentclass[a4paper,11pt]{article}
\usepackage{amsfonts}
\usepackage{amsmath}
\usepackage{amssymb}
\usepackage{geometry}

\newtheorem{lemma}{Lemma}

\newtheorem{proposition}{Proposition}
\newtheorem{remark}{Remark}

\newenvironment{proof}{\noindent\textbf{Proof:} }{\indent\hfill\rule{0.5em}{0.5em}}
\input{tcilatex}

\begin{document}

\title{Electroweak interaction without projection operators using
complexified octonions}
\author{John Fredsted\thanks{%
physics@johnfredsted.dk} \\
Soeskraenten 22, Stavtrup, DK-8260 Viby J., Denmark}
\maketitle

\begin{abstract}
Using complexified octonions, a formalism seemingly capable of describing
the coupling of spinors to the electroweak force with\textit{out} projection
operators is presented.
\end{abstract}

\section{Introduction}

The weak interaction violates parity maximally, coupling exclusively to the
left-handed components of Dirac four-spinors. In the standard model of
particle physics, this fact of Nature is mathematically modelled by \textit{%
putting in by hand} projection operators, $\frac{1}{2}\left( 1\pm \gamma
_{5}\right) $, leaving entirely unaddressed the question as to whether there
might exist some underlying physical or mathematical mechanism that might
explain this fact.

This paper will present a formalism seemingly capable of describing the
coupling of spinors to the electroweak force with\textit{out} projection
operators. The formalism will \textit{not} be given in the domain of
(associative) matrix calculus, a class to which the Dirac formalism and,
indeed, the entire standard model of particle physics belong; rather, it
will be given in the domain of (non-associative) complexified octonions.

\section{Notation}

\begin{itemize}
\item The sets of real numbers, complex numbers, quaternions, and octonions,
respectively, are denoted $\mathbb{R}$, $\mathbb{C}$, $\mathbb{H}$, and $%
\mathbb{O}$, as customarily.

\item The set of complexified octonions is denoted $\mathbb{C}\otimes 
\mathbb{O}$. Complex numbers and octonions are assumed to commute, i.e., $%
\mathbb{C}\otimes \mathbb{O}\equiv \mathbb{O}\otimes \mathbb{C}$.

\item The real- and imaginary parts of $\mathbb{C}$ are denoted $\func{Re}%
\left( \mathbb{C}\right) \cong \mathbb{R}$ and $\func{Im}\left( \mathbb{C}%
\right) \cong \mathrm{i}\mathbb{R}$, respectively, as customarily.

\item For any set $\mathbb{C}\otimes \mathbb{S}\subseteq \mathbb{C}\otimes 
\mathbb{O}$, usual complex conjugation acts as $\mathbb{C}\otimes \mathbb{%
S\rightarrow C}^{\ast }\otimes \mathbb{S}$, and usual octonionic conjugation 
\cite{Okubo,Springer and Veldkamp} acts as $\mathbb{C}\otimes \mathbb{%
S\rightarrow C}\otimes \overline{\mathbb{S}}$. In conjunction, these two
conjugations may be used to write $\mathbb{C}\otimes \mathbb{S}=\left( 
\mathbb{C}\otimes \mathbb{S}\right) ^{+}\cup \left( \mathbb{C}\otimes 
\mathbb{S}\right) ^{-}$, where%
\begin{equation*}
\left( \mathbb{C}\otimes \mathbb{S}\right) ^{\pm }\equiv \left\{ x\in 
\mathbb{C}\otimes \mathbb{S}\left\vert \overline{x}^{\ast }=\pm x\right.
\right\} .
\end{equation*}

\item For any set $\mathbb{C}\otimes \mathbb{S}\subseteq \mathbb{C}\otimes 
\mathbb{O}$, the scalar- and vector parts of $\mathbb{C}\otimes \mathbb{S}$,
denoted $\mathrm{Scal}\left( \mathbb{C}\otimes \mathbb{S}\right) $ and $%
\mathrm{Vec}\left( \mathbb{C}\otimes \mathbb{S}\right) $, respectively, are
defined as%
\begin{eqnarray*}
\mathrm{Scal}\left( \mathbb{C}\otimes \mathbb{S}\right) &\equiv &\left\{
x\in \mathbb{C}\otimes \mathbb{S}\left\vert \overline{x}=+x\right. \right\}
\equiv \mathbb{C}\otimes \mathrm{Scal}\left( \mathbb{S}\right) , \\
\mathrm{Vec}\left( \mathbb{C}\otimes \mathbb{S}\right) &\equiv &\left\{ x\in 
\mathbb{C}\otimes \mathbb{S}\left\vert \overline{x}=-x\right. \right\}
\equiv \mathbb{C}\otimes \mathrm{Vec}\left( \mathbb{S}\right) ,
\end{eqnarray*}%
so that $x=\mathrm{Scal}\left( x\right) +\mathrm{Vec}\left( x\right) $, for
any $x\in \mathbb{C}\otimes \mathbb{O}$.

\item The inner product, or scalar product, $\left\langle \cdot ,\cdot
\right\rangle :\left( \mathbb{C}\otimes \mathbb{O}\right) ^{2}\mathbb{%
\rightarrow C}$ is defined as \cite[Eq. (5)]{Dundarer and Gursey}%
\begin{equation*}
2\left\langle x,y\right\rangle \equiv x\overline{y}+y\overline{x}\equiv 
\overline{x}y+\overline{y}x.
\end{equation*}%
Note that $\left\langle x,y\right\rangle \equiv \left\langle
y,x\right\rangle $ and $\left\langle x,y\right\rangle \equiv \left\langle 
\overline{x},\overline{y}\right\rangle $.

\item The associator $\left[ \cdot ,\cdot ,\cdot \right] :\left( \mathbb{C}%
\otimes \mathbb{O}\right) ^{3}\mathbb{\rightarrow C}$ is defined as \cite[%
Eq. (12)]{Dundarer and Gursey}%
\begin{equation*}
\left[ x,y,z\right] \equiv \left( xy\right) z-x\left( yz\right) .
\end{equation*}

\item Let $\mathbb{C}\otimes \mathbb{A}\cong \mathbb{C}\otimes \mathbb{H}$
(with the choice of name referring to Associative) denote a specific, but
otherwise arbitrary, embedding (of which there are numerous \cite{Baez}) of $%
\mathbb{C}\otimes \mathbb{H}$ into $\mathbb{C}\otimes \mathbb{O}$, and let $%
\mathbb{C}\otimes \mathbb{B}\equiv \left( \mathbb{C}\otimes \mathbb{O}%
\right) \setminus \left( \mathbb{C}\otimes \mathbb{A}\right) $ denote its
non-associative complement.

\item Spacetime is assumed globally Minkowskian, with metric $\eta _{\mu \nu
}\equiv \left[ \mathrm{diag}\left( -1,1,1,1\right) \right] _{\mu \nu }$ and
parametrized by Cartesian coordinates $x^{\mu }$.

\item A \textit{Lorentz invariant} basis (over $\mathbb{C}$) for $\mathbb{C}%
\otimes \mathbb{A}$ is taken to be $\mathrm{e}_{\mu }\equiv \left( \mathrm{i}%
,\mathrm{e}_{1},\mathrm{e}_{2},\mathrm{e}_{3}\right) \in \left( \mathbb{C}%
\otimes \mathbb{A}\right) ^{-}$ [note Roman typed letters], where $\mathrm{i}
$ is the standard complex imaginary unit, and $\mathrm{e}_{i}\in \mathrm{Vec}%
\left( \mathbb{A}\right) $ are the imaginary quaternion units obeying $%
\mathrm{e}_{i}\mathrm{e}_{j}=-\delta _{ij}+\varepsilon _{ij}{}^{k}\mathrm{e}%
_{k}$, where $\varepsilon _{ijk}$ is the Levi-Civita pseudo-tensor with $%
\varepsilon _{123}=+1$. Note that $\overline{\mathrm{e}}_{\mu }^{\ast }=-%
\mathrm{e}_{\mu }$ and that $\left\langle \mathrm{e}_{\mu },\mathrm{e}_{\nu
}\right\rangle =\eta _{\mu \nu }$.

\item Let $\mathrm{M}_{n}\left( \mathbb{D}\right) $ denote the set of $n$%
-dimensional square matrices over some associative domain $\mathbb{D}$.
\end{itemize}

\section{Lorentz transformations}

Define $S_{\mu \nu }\in \mathbb{C}\otimes \mathrm{Vec}\left( \mathbb{A}%
\right) $ and $V_{\mu \nu }\in \mathrm{M}_{4}\left( \mathrm{i}\mathbb{R}%
\right) $, respectively, by%
\begin{eqnarray}
+4\mathrm{i}S_{\mu \nu } &\equiv &\mathrm{e}_{\mu }\overline{\mathrm{e}}%
_{\nu }-\mathrm{e}_{\nu }\overline{\mathrm{e}}_{\mu }\Leftrightarrow
\label{Eq:SGensL} \\
-4\mathrm{i}S_{\mu \nu }^{\ast } &\equiv &\overline{\mathrm{e}}_{\mu }%
\mathrm{e}_{\nu }-\overline{\mathrm{e}}_{\nu }\mathrm{e}_{\mu },
\label{Eq:SGensR}
\end{eqnarray}%
using for the bi-implication the property $\overline{\mathrm{e}}_{\mu
}^{\ast }=-\mathrm{e}_{\mu }$; and%
\begin{equation*}
\mathrm{i}\left( V_{\mu \nu }\right) ^{\rho }{}_{\sigma }\equiv \delta _{\mu
}^{\rho }\eta _{\nu \sigma }-\delta _{\nu }^{\rho }\eta _{\mu \sigma }.
\end{equation*}%
As is well-known, $V_{\mu \nu }$ are the generators of the vector
representation of the Lorentz group. Lemma \ref{Lemma:LorentzAlgebraS}
establishes that $S_{\mu \nu }$ are the generators of one of the two spinor
representations of the Lorentz group; the other spinor representation is
generated by $-S_{\mu \nu }^{\ast }$. By exponentiation, Lemma \ref%
{Lemma:DoubleCover} implies that%
\begin{equation}
\overline{\Lambda }_{S}^{\ast }\overline{\mathrm{e}}^{\rho }\Lambda
_{S}=\left( \Lambda _{V}\right) ^{\rho }{}_{\sigma }\overline{\mathrm{e}}%
^{\sigma },  \label{Eq:DoubleCoverFinite}
\end{equation}%
where $\Lambda _{S}\in \mathbb{C}\otimes \mathbb{A}$ and $\Lambda _{V}\in 
\mathrm{M}_{4}\left( \mathbb{R}\right) $, respectively, are defined by%
\begin{eqnarray*}
\Lambda _{S} &\equiv &\exp \left( -\frac{\mathrm{i}}{2}\theta ^{\mu \nu
}S_{\mu \nu }\right) , \\
\Lambda _{V} &\equiv &\exp \left( -\frac{\mathrm{i}}{2}\theta ^{\mu \nu
}V_{\mu \nu }\right) ,
\end{eqnarray*}%
with $\theta ^{\mu \nu }=-\theta ^{\nu \mu }\in \mathbb{R}$.

\begin{remark}
\normalfont Eq. (\ref{Eq:DoubleCoverFinite}) is the analogue of the relation 
$\Lambda _{1/2}^{-1}\gamma ^{\mu }\Lambda _{1/2}=\left( \Lambda _{V}\right)
^{\mu }{}_{\nu }\gamma ^{\nu }$ in the standard Dirac formalism \cite[Eq.
(3.29)]{Peskin and Schroeder}, where $\Lambda _{S}\equiv \exp \left( -\frac{%
\mathrm{i}}{2}\theta ^{\mu \nu }\sigma _{\mu \nu }\right) $ and $4\mathrm{i}%
\sigma _{\mu \nu }\equiv \left[ \gamma _{\mu },\gamma _{\nu }\right] $, with 
$\gamma _{\mu }$ being the standard flat spacetime Dirac gamma matrices.
There is a difference, though: $\overline{\Lambda }_{S}^{\ast }$ does 
\textit{only} equal $\Lambda _{S}^{-1}$ for spatial rotations; for pure
boosts, for instance, it equals $\Lambda _{S}$ itself.
\end{remark}

Consider spinor fields $\alpha =\alpha \left( x\right) \in \mathbb{C}\otimes 
\mathbb{A}$ and $\beta =\beta \left( x\right) \in \mathbb{C}\otimes \mathbb{B%
}$ transforming under Lorentz transformations as (suppressing the Lorentz
transformation of the arguments)%
\begin{eqnarray}
\alpha ^{\prime } &=&\Lambda _{S}\alpha ,  \label{Eq:TransLorentzAlpha} \\
\beta ^{\prime } &=&\overline{\Lambda }_{S}^{\ast }\beta \equiv \beta
\Lambda _{S}^{\ast },  \label{Eq:TransLorentzBeta}
\end{eqnarray}%
using Eq. (\ref{Eq:abIdentity}). These transformations are compatible, as
they should be, with group composition (from the left): $\alpha ^{\prime
\prime }=\Lambda _{2}\left( \Lambda _{1}\alpha \right) =\Lambda _{3}\alpha $
and $\beta ^{\prime \prime }=\overline{\Lambda }_{2}^{\ast }\left( \overline{%
\Lambda }_{1}^{\ast }\beta \right) =\overline{\Lambda }_{3}^{\ast }\beta $,
where $\Lambda _{3}\equiv \Lambda _{2}\Lambda _{1}$ [suppressing here the $S$
in $\Lambda _{S}$]. The relation for $\alpha $ follows by associativity, as
all quantities belong to $\mathbb{C}\otimes \mathbb{A}$; the relation for $%
\beta $ follows from Eq. (\ref{Eq:aabIdentity}).

\begin{remark}
\label{Remark:NonassocBetaRep}\normalfont Note that the representation in
which $\beta $ transforms, Eq. (\ref{Eq:TransLorentzBeta}), exists due to
the \textit{non}associative identity $\left( aa^{\prime }\right) b\equiv
a^{\prime }\left( ab\right) $, Eq. (\ref{Eq:aabIdentity}), and thus can%
\textit{not} correspond to any representation encountered in usual
associative formalisms.
\end{remark}

\begin{proposition}
\label{Prop:LorentInvariantIPs}The quantities $\left\langle \alpha ^{\ast },%
\overline{\mathrm{e}}^{\rho }\partial _{\rho }\alpha \right\rangle $ and $%
\left\langle \beta ^{\ast },\overline{\mathrm{e}}^{\rho }\partial _{\rho
}\beta \right\rangle $ are separately globally Lorentz invariant.
\end{proposition}

\begin{proof}
Due to associativity, all quantities belonging to $\mathbb{C}\otimes \mathbb{%
A}$, the calculation for $\left\langle \alpha ^{\ast },\overline{\mathrm{e}}%
^{\rho }\partial _{\rho }\alpha \right\rangle $ is straightforward:%
\begin{eqnarray*}
\left\langle \alpha ^{\ast },\overline{\mathrm{e}}^{\rho }\partial _{\rho
}\alpha \right\rangle ^{\prime } &=&\left( \Lambda _{V}^{-1}\right) ^{\sigma
}{}_{\rho }\left\langle \left( \Lambda _{S}\alpha \right) ^{\ast },\overline{%
\mathrm{e}}^{\rho }\partial _{\sigma }\left( \Lambda _{S}\alpha \right)
\right\rangle \\
&=&\left( \Lambda _{V}^{-1}\right) ^{\sigma }{}_{\rho }\left\langle \Lambda
_{S}^{\ast }\alpha ^{\ast },\overline{\mathrm{e}}^{\rho }\Lambda
_{S}\partial _{\sigma }\alpha \right\rangle \\
&=&\left( \Lambda _{V}^{-1}\right) ^{\sigma }{}_{\rho }\left\langle \alpha
^{\ast },\overline{\Lambda }_{S}^{\ast }\overline{\mathrm{e}}^{\rho }\Lambda
_{S}\partial _{\sigma }\alpha \right\rangle \\
&=&\left( \Lambda _{V}^{-1}\right) ^{\sigma }{}_{\rho }\left( \Lambda
_{V}\right) ^{\rho }{}_{\tau }\left\langle \alpha ^{\ast },\overline{\mathrm{%
e}}^{\tau }\partial _{\sigma }\alpha \right\rangle \\
&=&\left\langle \alpha ^{\ast },\overline{\mathrm{e}}^{\rho }\partial _{\rho
}\alpha \right\rangle ,
\end{eqnarray*}%
using Eqs. (\ref{Eq:DoubleCoverFinite}) and (\ref{Eq:ipMoveLL}). Due to
nonassociativity, the calculation for $\left\langle \beta ^{\ast },\overline{%
\mathrm{e}}^{\rho }\partial _{\rho }\beta \right\rangle $ requires a bit
more care:%
\begin{eqnarray*}
\left\langle \beta ^{\ast },\overline{\mathrm{e}}^{\rho }\partial _{\rho
}\beta \right\rangle ^{\prime } &=&\left( \Lambda _{V}^{-1}\right) ^{\sigma
}{}_{\rho }\left\langle \left( \overline{\Lambda }_{S}^{\ast }\beta \right)
^{\ast },\overline{\mathrm{e}}^{\rho }\partial _{\sigma }\left( \overline{%
\Lambda }_{S}^{\ast }\beta \right) \right\rangle \\
&=&\left( \Lambda _{V}^{-1}\right) ^{\sigma }{}_{\rho }\left\langle \left( 
\overline{\Lambda }_{S}^{\ast }\beta \right) ^{\ast },\overline{\mathrm{e}}%
^{\rho }\left( \overline{\Lambda }_{S}^{\ast }\partial _{\sigma }\beta
\right) \right\rangle \\
&=&\left( \Lambda _{V}^{-1}\right) ^{\sigma }{}_{\rho }\left\langle 
\overline{\Lambda }_{S}\beta ^{\ast },\left( \overline{\Lambda }_{S}^{\ast }%
\overline{\mathrm{e}}^{\rho }\right) \partial _{\sigma }\beta \right\rangle
\\
&=&\left( \Lambda _{V}^{-1}\right) ^{\sigma }{}_{\rho }\left\langle \beta
^{\ast },\Lambda _{S}\left[ \left( \overline{\Lambda }_{S}^{\ast }\overline{%
\mathrm{e}}^{\rho }\right) \partial _{\sigma }\beta \right] \right\rangle \\
&=&\left( \Lambda _{V}^{-1}\right) ^{\sigma }{}_{\rho }\left\langle \beta
^{\ast },\left( \overline{\Lambda }_{S}^{\ast }\overline{\mathrm{e}}^{\rho
}\Lambda _{S}\right) \partial _{\sigma }\beta \right\rangle \\
&=&\left( \Lambda _{V}^{-1}\right) ^{\sigma }{}_{\rho }\left( \Lambda
_{V}\right) ^{\rho }{}_{\tau }\left\langle \beta ^{\ast },\overline{\mathrm{e%
}}^{\tau }\partial _{\sigma }\beta \right\rangle \\
&=&\left\langle \beta ^{\ast },\overline{\mathrm{e}}^{\rho }\partial _{\rho
}\beta \right\rangle ,
\end{eqnarray*}%
using Eqs. (\ref{Eq:DoubleCoverFinite}), (\ref{Eq:ipMoveLL}) and (\ref%
{Eq:aabIdentity}).
\end{proof}

\section{Gauge transformations}

An interesting question is whether the Lorentz invariant quantities $%
\left\langle \alpha ^{\ast },\overline{\mathrm{e}}^{\rho }\partial _{\rho
}\alpha \right\rangle $ and $\left\langle \beta ^{\ast },\overline{\mathrm{e}%
}^{\rho }\partial _{\rho }\beta \right\rangle $, as considered in
Proposition \ref{Prop:LorentInvariantIPs}, possess some local gauge freedom.

Turning attention first to $\alpha $, consider a gauge transformation of the
form $\alpha ^{\prime }=\alpha U^{-1}$, where $U\equiv \exp \left( u\right)
\in \mathbb{C}\otimes \mathbb{A}$ (note exponentiation of $u\in \mathbb{C}%
\otimes \mathbb{A}$). This gauge transformation is compatible with group
composition (from the left) in the sense that $\alpha ^{\prime \prime
}=\left( \alpha U_{1}^{-1}\right) U_{2}^{-1}=\alpha U_{3}^{-1}$, where $%
U_{3}\equiv U_{2}U_{1}$; and it commutes with the Lorentz transformation $%
\alpha ^{\prime }=\Lambda _{S}\alpha $ (different primes not to be confused)
due to associativity, $0=\left[ \Lambda _{S},\alpha ,U^{-1}\right] \equiv
\left( \Lambda _{S}\alpha \right) U^{-1}-\Lambda _{S}\left( \alpha
U^{-1}\right) $, thus respecting the Coleman-Mandula theorem \cite[Sect. 24.B%
]{Weinberg}.

\begin{proposition}
The quantity $\left\langle \alpha ^{\ast },\overline{\mathrm{e}}^{\rho
}\partial _{\rho }\alpha \right\rangle $ is invariant under the global gauge
transformation $\alpha ^{\prime }=\alpha U^{-1}$, if and only if $\overline{U%
}^{\ast }=U^{-1}$.
\end{proposition}

\begin{proof}
Due to associativity, all quantities belonging to $\mathbb{C}\otimes \mathbb{%
A}$, the calculation is straightforward:%
\begin{eqnarray*}
\left\langle \alpha ^{\ast },\overline{\mathrm{e}}^{\rho }\partial _{\rho
}\alpha \right\rangle ^{\prime } &=&\left\langle \left( \alpha U^{-1}\right)
^{\ast },\overline{\mathrm{e}}^{\rho }\partial _{\rho }\left( \alpha
U^{-1}\right) \right\rangle \\
&=&\left\langle \alpha ^{\ast }\left( U^{\ast }\right) ^{-1},\overline{%
\mathrm{e}}^{\rho }\left( \partial _{\rho }\alpha \right) U^{-1}\right\rangle
\\
&=&\left\langle \alpha ^{\ast },\overline{\mathrm{e}}^{\rho }\left( \partial
_{\rho }\alpha \right) U^{-1}\left( \overline{U}^{\ast }\right)
^{-1}\right\rangle \\
&=&\left\langle \alpha ^{\ast },\overline{\mathrm{e}}^{\rho }\left( \partial
_{\rho }\alpha \right) \left( \overline{U}^{\ast }U\right)
^{-1}\right\rangle ,
\end{eqnarray*}%
using Eq. (\ref{Eq:ipMoveLR}). Invariance requires $\overline{U}^{\ast }U=1$.
\end{proof}

\vspace{0pt}As $U\equiv \exp \left( u\right) \in \mathbb{C}\otimes \mathbb{A}
$, the condition $\overline{U}^{\ast }=U^{-1}$ is equivalent to $\overline{u}%
^{\ast }=-u\Leftrightarrow u\in \left( \mathbb{C}\otimes \mathbb{A}\right)
^{-}$. As $\left( \mathbb{C}\otimes \mathbb{A}\right) ^{-}$ is isomorphic to
the Lie algebra $\mathfrak{su}\left( 2\right) \oplus \mathfrak{u}\left(
1\right) $, the gauge transformation $\alpha ^{\prime }=\alpha U^{-1}=\alpha
\exp \left( -u\right) $ therefore corresponds to a \textit{global} $\mathrm{%
SU}\left( 2\right) \times \mathrm{U}\left( 1\right) $ gauge symmetry of $%
\left\langle \alpha ^{\ast },\overline{\mathrm{e}}^{\rho }\partial _{\rho
}\alpha \right\rangle $. In order to lift this global symmetry to a \textit{%
local} one, introduce an $\mathrm{SU}\left( 2\right) \times \mathrm{U}\left(
1\right) $ gauge connection $W_{\rho }\in \left( \mathbb{C}\otimes \mathbb{A}%
\right) ^{-}$ transforming standardly as%
\begin{equation}
W_{\rho }^{\prime }=UW_{\rho }U^{-1}-\left( \partial _{\rho }U\right) U^{-1}
\label{Eq:TransLorentzW}
\end{equation}%
under local $\mathrm{SU}\left( 2\right) \times \mathrm{U}\left( 1\right) $
gauge transformations.

\begin{proposition}
\label{Prop:CovDerAlpha}The $\mathrm{SU}\left( 2\right) \times \mathrm{U}%
\left( 1\right) $ covariant derivative of $\alpha $ is given by%
\begin{equation}
D_{\rho }\alpha \equiv \partial _{\rho }\alpha -\alpha W_{\rho }.
\label{Eq:CovDerAlpha}
\end{equation}
\end{proposition}

\begin{proof}
Due to associativity, all quantities belonging to $\mathbb{C}\otimes \mathbb{%
A}$, the calculation is straightforward:%
\begin{eqnarray*}
\left( D_{\rho }\alpha \right) ^{\prime } &=&\partial _{\rho }\left( \alpha
U^{-1}\right) -\left( \alpha U^{-1}\right) \left[ UW_{\rho }U^{-1}-\left(
\partial _{\rho }U\right) U^{-1}\right] \\
&=&\left( \partial _{\rho }\alpha \right) U^{-1}+\alpha \left( \partial
_{\rho }U^{-1}\right) -\alpha \left( W_{\rho }U^{-1}+\partial _{\rho
}U^{-1}\right) \\
&=&\left( \partial _{\rho }\alpha \right) U^{-1}-\alpha W_{\rho }U^{-1} \\
&=&\left( D_{\rho }\alpha \right) U^{-1},
\end{eqnarray*}%
using $\partial _{\rho }U^{-1}=-U^{-1}\left( \partial _{\rho }U\right)
U^{-1} $.
\end{proof}

\vspace{\baselineskip}\vspace{0pt}Due to this Proposition, the quantity $%
\left\langle \alpha ^{\ast },\overline{\mathrm{e}}^{\rho }D_{\rho }\alpha
\right\rangle $ is \textit{locally} $\mathrm{SU}\left( 2\right) \times 
\mathrm{U}\left( 1\right) $ invariant for the very same reason that $%
\left\langle \alpha ^{\ast },\overline{\mathrm{e}}^{\rho }\partial _{\rho
}\alpha \right\rangle $ is \textit{globally} $\mathrm{SU}\left( 2\right)
\times \mathrm{U}\left( 1\right) $ invariant, as previously proved.

Turning attention next to $\beta $, an interesting question is how it can be
coupled to the gauge field $W_{\rho }$ just introduced in connection with $%
\alpha $.

\begin{proposition}
\label{Prop:CovDerBeta}The most general $\mathrm{SU}\left( 2\right) \times 
\mathrm{U}\left( 1\right) $ covariant derivative of $\beta $ is given by%
\begin{eqnarray}
D_{\rho }\beta &=&\partial _{\rho }\beta +\frac{r}{2}\left( \beta W_{\rho
}+W_{\rho }\beta \right)  \notag \\
&\equiv &\partial _{\rho }\beta +\frac{r}{2}\left( \beta W_{\rho }+\beta 
\overline{W}_{\rho }\right)  \notag \\
&\equiv &\partial _{\rho }\beta +r\beta \mathrm{Scal}\left( W_{\rho }\right)
,  \label{Eq:CovDerBeta}
\end{eqnarray}%
where $r\in \mathbb{R}$. All three expressions are equivalent due to Eq. (%
\ref{Eq:abIdentity}) and the identity $2\mathrm{Scal}\left( x\right) \equiv
x+\overline{x}$.
\end{proposition}

\begin{proof}
Let $D_{\rho }\equiv \partial _{\rho }\beta +V_{\rho }\left( \beta \right) $%
. The most general $V_{\rho }\left( \beta \right) $ is given by a linear
combination over $\mathbb{R}$ of the eight possible terms%
\begin{eqnarray*}
&&\beta W_{\rho },\beta W_{\rho }^{\ast },\beta \overline{W}_{\rho },\beta 
\overline{W}_{\rho }^{\ast }, \\
&&W_{\rho }\beta ,W_{\rho }^{\ast }\beta ,\overline{W}_{\rho }\beta ,%
\overline{W}_{\rho }^{\ast }\beta .
\end{eqnarray*}%
Due to the identity $ab\equiv b\overline{a}$, Eq. (\ref{Eq:abIdentity}), and
the property $\overline{W}_{\rho }^{\ast }=-W_{\rho }$, following from $%
W_{\rho }\in \left( \mathbb{C}\otimes \mathbb{A}\right) ^{-}$, such a linear
combination can be rewritten as $V_{\rho }\left( \beta \right) =r_{1}W_{\rho
}\beta +r_{2}\beta W_{\rho }$, where $r_{1},r_{2}\in \mathbb{R}$. Under
(global) Lorentz transformations, this expression transforms as%
\begin{eqnarray*}
\left[ V_{\rho }\left( \beta \right) \right] ^{\prime } &=&\left( \Lambda
_{V}^{-1}\right) ^{\sigma }{}_{\rho }\left[ r_{1}W_{\sigma }\left( \overline{%
\Lambda }_{S}^{\ast }\beta \right) +r_{2}\left( \overline{\Lambda }%
_{S}^{\ast }\beta \right) W_{\sigma }\right] \\
&=&\left( \Lambda _{V}^{-1}\right) ^{\sigma }{}_{\rho }\left[ r_{1}W_{\sigma
}\left( \overline{\Lambda }_{S}^{\ast }\beta \right) +r_{2}\left( \beta
\Lambda _{S}^{\ast }\right) W_{\sigma }\right] \\
&=&\left( \Lambda _{V}^{-1}\right) ^{\sigma }{}_{\rho }\left[ r_{1}\left( 
\overline{\Lambda }_{S}^{\ast }W_{\sigma }\right) \beta +r_{2}\beta \left(
W_{\sigma }\Lambda _{S}^{\ast }\right) \right] \\
&=&\left( \Lambda _{V}^{-1}\right) ^{\sigma }{}_{\rho }\left[ r_{1}\left( 
\overline{\Lambda }_{S}^{\ast }W_{\sigma }\right) \beta +r_{2}\left( 
\overline{\Lambda }_{S}^{\ast }\overline{W}_{\sigma }\right) \beta \right] \\
&=&\left( \Lambda _{V}^{-1}\right) ^{\sigma }{}_{\rho }\left[ \overline{%
\Lambda }_{S}^{\ast }\left( r_{1}W_{\sigma }+r_{2}\overline{W}_{\sigma
}\right) \right] \beta ,
\end{eqnarray*}%
using Eqs. (\ref{Eq:abIdentity}) and (\ref{Eq:aabIdentity})-(\ref%
{Eq:baaIdentity}). For $D_{\rho }\beta $ to transform under global Lorentz
transformations as $\left( D_{\rho }\beta \right) ^{\prime }=\left( \Lambda
_{V}^{-1}\right) ^{\sigma }{}_{\rho }\overline{\Lambda }_{S}^{\ast }\left(
D_{\sigma }\beta \right) $, as it should, the following condition must thus
hold:%
\begin{eqnarray*}
\left[ V_{\rho }\left( \beta \right) \right] ^{\prime } &=&\left( \Lambda
_{V}^{-1}\right) ^{\sigma }{}_{\rho }\overline{\Lambda }_{S}^{\ast
}V_{\sigma }\left( \beta \right) \\
&=&\left( \Lambda _{V}^{-1}\right) ^{\sigma }{}_{\rho }\overline{\Lambda }%
_{S}^{\ast }\left[ r_{1}W_{\sigma }\beta +r_{2}\beta W_{\sigma }\right] \\
&=&\left( \Lambda _{V}^{-1}\right) ^{\sigma }{}_{\rho }\overline{\Lambda }%
_{S}^{\ast }\left[ \left( r_{1}W_{\sigma }+r_{2}\overline{W}_{\sigma
}\right) \beta \right] ,
\end{eqnarray*}%
using Eq. (\ref{Eq:abIdentity}). The two expressions for $\left[ V_{\rho
}\left( \beta \right) \right] ^{\prime }$ can only agree if%
\begin{eqnarray*}
0 &=&\left[ \overline{\Lambda }_{S}^{\ast }\left( r_{1}W_{\sigma }+r_{2}%
\overline{W}_{\sigma }\right) \right] \beta -\overline{\Lambda }_{S}^{\ast }%
\left[ \left( r_{1}W_{\sigma }+r_{2}\overline{W}_{\sigma }\right) \beta %
\right] \\
&\equiv &\left[ \overline{\Lambda }_{S}^{\ast },r_{1}W_{\sigma }+r_{2}%
\overline{W}_{\sigma },\beta \right] ,
\end{eqnarray*}%
i.e., if the three quantities in question associate. For general $\beta $
and $\Lambda _{S}$, this is only possible if%
\begin{eqnarray*}
\mathbb{C} &\ni &r_{1}W_{\sigma }+r_{2}\overline{W}_{\sigma } \\
&=&r_{1}\left[ \mathrm{Scal}\left( W_{\rho }\right) +\mathrm{Vec}\left(
W_{\rho }\right) \right] +r_{2}\left[ \mathrm{Scal}\left( W_{\rho }\right) -%
\mathrm{Vec}\left( W_{\rho }\right) \right] \\
&=&\left( r_{1}+r_{2}\right) \mathrm{Scal}\left( W_{\rho }\right) +\left(
r_{1}-r_{2}\right) \mathrm{Vec}\left( W_{\rho }\right) ,
\end{eqnarray*}%
which implies $r_{1}=r_{2}\equiv r/2$, thus concluding the proof.
\end{proof}

\vspace{\baselineskip}To wrap up matters nicely, the transformation of $%
\beta $ under local $\mathrm{SU}\left( 2\right) \times \mathrm{U}\left(
1\right) $ gauge transformations, or, effectively, local $\mathrm{U}\left(
1\right) $ gauge transformations, is now given.

\begin{proposition}
\label{Prop:TransGaugeBeta}The $\mathrm{U}\left( 1\right) $ local gauge
transformation of $\beta $ corresponding to Eq. (\ref{Eq:CovDerBeta}) is
given by%
\begin{equation*}
\beta ^{\prime }=\beta \exp \left( r\mathrm{Scal}\left( u\right) \right) ,
\end{equation*}%
where $u$, of course, is the parameter entering into the definition $U\equiv
\exp \left( u\right) $.
\end{proposition}

\begin{proof}
Due to associativity, all quantities belonging to $\mathbb{C}\otimes \mathbb{%
A}$, the calculation is straightforward:%
\begin{eqnarray*}
\left( D_{\rho }\beta \right) ^{\prime } &=&\partial _{\rho }\left[ \beta
\exp \left( r\mathrm{Scal}\left( u\right) \right) \right] +r\left[ \beta
\exp \left( r\mathrm{Scal}\left( u\right) \right) \right] \mathrm{Scal}%
\left( W_{\rho }^{\prime }\right) \\
&=&\left[ \partial _{\rho }\beta +\beta r\mathrm{Scal}\left( \partial _{\rho
}u\right) +r\beta \mathrm{Scal}\left( W_{\rho }^{\prime }\right) \right]
\exp \left( r\mathrm{Scal}\left( u\right) \right) \\
&=&\left[ \partial _{\rho }\beta +r\beta \mathrm{Scal}\left( W_{\rho
}^{\prime }+\partial _{\rho }u\right) \right] \exp \left( r\mathrm{Scal}%
\left( u\right) \right) \\
&=&\left[ \partial _{\rho }\beta +r\beta \mathrm{Scal}\left( W_{\rho
}\right) \right] \exp \left( r\mathrm{Scal}\left( u\right) \right) \\
&=&\left( D_{\rho }\beta \right) \exp \left( r\mathrm{Scal}\left( u\right)
\right) ,
\end{eqnarray*}%
using Lemma \ref{Lemma:ScalWW}, and the fact that $\exp \left( r\mathrm{Scal}%
\left( u\right) \right) \in \mathbb{C}$ and $\mathrm{Scal}\left( W_{\rho
}^{\prime }\right) \in \mathrm{i}\mathbb{R}$ commute and associate with
anything.
\end{proof}

\vspace{\baselineskip}Note that $\exp \left( r\mathrm{Scal}\left( u\right)
\right) \in \mathbb{C}$, as introduced in this Proposition, is indeed a $%
\mathrm{U}\left( 1\right) $ phase, as the following calculation shows:%
\begin{eqnarray*}
\exp \left( r\mathrm{Scal}\left( u\right) \right) ^{\ast } &=&\exp \left( r%
\mathrm{Scal}\left( u^{\ast }\right) \right) \\
&=&\exp \left( r\mathrm{Scal}\left( -\overline{u}\right) \right) \\
&=&\exp \left( r\mathrm{Scal}\left( -u\right) \right) \\
&=&\exp \left( -r\mathrm{Scal}\left( u\right) \right) \\
&=&\exp \left( r\mathrm{Scal}\left( u\right) \right) ^{-1},
\end{eqnarray*}%
using $\overline{u}^{\ast }=-u$. Thus the quantity $\left\langle \beta
^{\ast },\overline{\mathrm{e}}^{\rho }D_{\rho }\beta \right\rangle $ is 
\textit{locally} $\mathrm{U}\left( 1\right) $ invariant.

\section{Conclusion}

As $\mathrm{Scal}\left( W_{\rho }\right) \in \mathrm{i}\mathbb{R}$, because $%
W_{\rho }\in \left( \mathbb{C}\otimes \mathbb{A}\right) ^{-}\equiv \mathrm{i}%
\mathbb{R}\cup \mathrm{Vec}\left( \mathbb{A}\right) $, Propositions \ref%
{Prop:CovDerBeta} and \ref{Prop:TransGaugeBeta} say that $\beta $, unlike $%
\alpha $, can couple \textit{only} to the $\mathrm{U}\left( 1\right) $ part
of the $\mathrm{SU}\left( 2\right) \times \mathrm{U}\left( 1\right) $ gauge
connection $W_{\rho }$.

It is therefore tempting to assign to $\alpha $ left-handed $\mathrm{SU}%
\left( 2\right) $ \textit{doublets}, and to $\beta $ right-handed $\mathrm{SU%
}\left( 2\right) $ \textit{singlets}. For $\beta $, though, this is not
quite satisfactory as $\beta $ possesses four (complex) degrees of freedom
(as does $\alpha $), whereas an $\mathrm{SU}\left( 2\right) $ singlet
possesses only two (complex) degrees of freedom: it seems that some
truncation of the degrees of freedom of $\beta $ is required; note, however,
that this truncation can\textit{not} possibly be performed by some
projection operator, as in the present formalism the analogue of $\gamma
_{5}\equiv -\mathrm{i}\gamma ^{0}\gamma ^{1}\gamma ^{2}\gamma ^{3}$ in the
Dirac formalism seems trivial due to the relation $-\mathrm{ie}^{0}\overline{%
\mathrm{e}}^{1}\mathrm{e}^{2}\overline{\mathrm{e}}^{3}=1$.

Thus, the formalism here presented, crucially depending on nonassociativity,
as previously noted in Remark \ref{Remark:NonassocBetaRep}, seems capable of
describing the coupling of spinors to the electroweak force with\textit{out}
projection operators $\frac{1}{2}\left( 1\pm \gamma _{5}\right) $.

\section{Auxiliary material, I: Lorentz transformations}

\begin{lemma}
\label{Lemma:LorentzAlgebraS}The quantities $S_{\mu \nu }$, Eq. (\ref%
{Eq:SGensL}), obey the Lorentz algebra,%
\begin{equation*}
-\mathrm{i}\left[ S_{\mu \nu },S_{\rho \sigma }\right] =\eta _{\mu \rho
}S_{\nu \sigma }-\eta _{\mu \sigma }S_{\nu \rho }-\eta _{\nu \rho }S_{\mu
\sigma }+\eta _{\nu \sigma }S_{\mu \rho }.
\end{equation*}
\end{lemma}

\begin{proof}
By a calculation completely analogous to the proof that $-\frac{\mathrm{i}}{4%
}\left[ \gamma _{\mu },\gamma _{\nu }\right] $ in the standard Dirac
formalism obey the Lorentz algebra, using here $\mathrm{e}_{\mu }\overline{%
\mathrm{e}}_{\nu }+\mathrm{e}_{\nu }\overline{\mathrm{e}}_{\mu }=\overline{%
\mathrm{e}}_{\mu }\mathrm{e}_{\nu }+\overline{\mathrm{e}}_{\nu }\mathrm{e}%
_{\mu }=2\eta _{\mu \nu }$:%
\begin{eqnarray*}
\mathrm{e}_{\mu }\overline{\mathrm{e}}_{\nu }\mathrm{e}_{\rho }\overline{%
\mathrm{e}}_{\sigma } &=&\mathrm{e}_{\mu }\left( 2\eta _{\nu \rho }-%
\overline{\mathrm{e}}_{\rho }\mathrm{e}_{\nu }\right) \overline{\mathrm{e}}%
_{\sigma } \\
&=&2\eta _{\nu \rho }\mathrm{e}_{\mu }\overline{\mathrm{e}}_{\sigma }-%
\mathrm{e}_{\mu }\overline{\mathrm{e}}_{\rho }\left( 2\eta _{\nu \sigma }-%
\mathrm{e}_{\sigma }\overline{\mathrm{e}}_{\nu }\right) \\
&=&2\eta _{\nu \rho }\mathrm{e}_{\mu }\overline{\mathrm{e}}_{\sigma }-2\eta
_{\nu \sigma }\mathrm{e}_{\mu }\overline{\mathrm{e}}_{\rho }+\left( 2\eta
_{\mu \rho }-\mathrm{e}_{\rho }\overline{\mathrm{e}}_{\mu }\right) \mathrm{e}%
_{\sigma }\overline{\mathrm{e}}_{\nu } \\
&=&2\eta _{\nu \rho }\mathrm{e}_{\mu }\overline{\mathrm{e}}_{\sigma }-2\eta
_{\nu \sigma }\mathrm{e}_{\mu }\overline{\mathrm{e}}_{\rho }+2\eta _{\mu
\rho }\mathrm{e}_{\sigma }\overline{\mathrm{e}}_{\nu }-\mathrm{e}_{\rho
}\left( 2\eta _{\mu \sigma }-\overline{\mathrm{e}}_{\sigma }\mathrm{e}_{\mu
}\right) \overline{\mathrm{e}}_{\nu }\Rightarrow \\
\left[ \mathrm{e}_{\mu }\overline{\mathrm{e}}_{\nu },\mathrm{e}_{\rho }%
\overline{\mathrm{e}}_{\sigma }\right] &=&2\eta _{\nu \rho }\mathrm{e}_{\mu }%
\overline{\mathrm{e}}_{\sigma }-2\eta _{\nu \sigma }\mathrm{e}_{\mu }%
\overline{\mathrm{e}}_{\rho }+2\eta _{\mu \rho }\mathrm{e}_{\sigma }%
\overline{\mathrm{e}}_{\nu }-2\eta _{\mu \sigma }\mathrm{e}_{\rho }\overline{%
\mathrm{e}}_{\nu },
\end{eqnarray*}%
from which it follows that%
\begin{eqnarray*}
-16\left[ S_{\mu \nu },S_{\rho \sigma }\right] &=&\left[ \mathrm{e}_{\mu }%
\overline{\mathrm{e}}_{\nu },\mathrm{e}_{\rho }\overline{\mathrm{e}}_{\sigma
}\right] -\left[ \mathrm{e}_{\nu }\overline{\mathrm{e}}_{\mu },\mathrm{e}%
_{\rho }\overline{\mathrm{e}}_{\sigma }\right] -\left[ \mathrm{e}_{\mu }%
\overline{\mathrm{e}}_{\nu },\mathrm{e}_{\sigma }\overline{\mathrm{e}}_{\rho
}\right] +\left[ \mathrm{e}_{\nu }\overline{\mathrm{e}}_{\mu },\mathrm{e}%
_{\sigma }\overline{\mathrm{e}}_{\rho }\right] \\
&=&\left( 2\eta _{\nu \rho }\mathrm{e}_{\mu }\overline{\mathrm{e}}_{\sigma
}-2\eta _{\nu \sigma }\mathrm{e}_{\mu }\overline{\mathrm{e}}_{\rho }+2\eta
_{\mu \rho }\mathrm{e}_{\sigma }\overline{\mathrm{e}}_{\nu }-2\eta _{\mu
\sigma }\mathrm{e}_{\rho }\overline{\mathrm{e}}_{\nu }\right) \\
&&-\left( 2\eta _{\mu \rho }\mathrm{e}_{\nu }\overline{\mathrm{e}}_{\sigma
}-2\eta _{\mu \sigma }\mathrm{e}_{\nu }\overline{\mathrm{e}}_{\rho }+2\eta
_{\nu \rho }\mathrm{e}_{\sigma }\overline{\mathrm{e}}_{\mu }-2\eta _{\nu
\sigma }\mathrm{e}_{\rho }\overline{\mathrm{e}}_{\mu }\right) \\
&&-\left( 2\eta _{\nu \sigma }\mathrm{e}_{\mu }\overline{\mathrm{e}}_{\rho
}-2\eta _{\nu \rho }\mathrm{e}_{\mu }\overline{\mathrm{e}}_{\sigma }+2\eta
_{\mu \sigma }\mathrm{e}_{\rho }\overline{\mathrm{e}}_{\nu }-2\eta _{\mu
\rho }\mathrm{e}_{\sigma }\overline{\mathrm{e}}_{\nu }\right) \\
&&+\left( 2\eta _{\mu \sigma }\mathrm{e}_{\nu }\overline{\mathrm{e}}_{\rho
}-2\eta _{\mu \rho }\mathrm{e}_{\nu }\overline{\mathrm{e}}_{\sigma }+2\eta
_{\nu \sigma }\mathrm{e}_{\rho }\overline{\mathrm{e}}_{\mu }-2\eta _{\nu
\rho }\mathrm{e}_{\sigma }\overline{\mathrm{e}}_{\mu }\right) \\
&=&4\eta _{\nu \rho }\left( \mathrm{e}_{\mu }\overline{\mathrm{e}}_{\sigma }-%
\mathrm{e}_{\sigma }\overline{\mathrm{e}}_{\mu }\right) -4\eta _{\nu \sigma
}\left( \mathrm{e}_{\mu }\overline{\mathrm{e}}_{\rho }-\mathrm{e}_{\rho }%
\overline{\mathrm{e}}_{\mu }\right) \\
&&+4\eta _{\mu \rho }\left( \mathrm{e}_{\sigma }\overline{\mathrm{e}}_{\nu }-%
\mathrm{e}_{\nu }\overline{\mathrm{e}}_{\sigma }\right) -4\eta _{\mu \sigma
}\left( \mathrm{e}_{\rho }\overline{\mathrm{e}}_{\nu }-\mathrm{e}_{\nu }%
\overline{\mathrm{e}}_{\rho }\right) ,
\end{eqnarray*}%
from which the result readily follows.
\end{proof}

\begin{lemma}
\label{Lemma:DoubleCover}The following identity holds:%
\begin{equation*}
S_{\mu \nu }^{\ast }\overline{\mathrm{e}}^{\rho }+\overline{\mathrm{e}}%
^{\rho }S_{\mu \nu }=\left( V_{\mu \nu }\right) ^{\rho }{}_{\sigma }%
\overline{\mathrm{e}}^{\sigma }.
\end{equation*}
\end{lemma}

\begin{proof}
By direct calculation, using $\mathrm{e}_{\mu }\overline{\mathrm{e}}_{\nu }+%
\mathrm{e}_{\nu }\overline{\mathrm{e}}_{\mu }=\overline{\mathrm{e}}_{\mu }%
\mathrm{e}_{\nu }+\overline{\mathrm{e}}_{\nu }\mathrm{e}_{\mu }=2\eta _{\mu
\nu }$:%
\begin{eqnarray*}
\overline{\mathrm{e}}_{\mu }\mathrm{e}_{\nu }\overline{\mathrm{e}}^{\rho }
&=&\overline{\mathrm{e}}_{\mu }\left( 2\delta _{\nu }^{\rho }-\mathrm{e}%
^{\rho }\overline{\mathrm{e}}_{\nu }\right) \\
&=&2\delta _{\nu }^{\rho }\overline{\mathrm{e}}_{\mu }-\overline{\mathrm{e}}%
_{\mu }\mathrm{e}^{\rho }\overline{\mathrm{e}}_{\nu } \\
&=&2\delta _{\nu }^{\rho }\overline{\mathrm{e}}_{\mu }-\left( 2\delta _{\mu
}^{\rho }-\overline{\mathrm{e}}^{\rho }\mathrm{e}_{\mu }\right) \overline{%
\mathrm{e}}_{\nu } \\
&=&2\delta _{\nu }^{\rho }\overline{\mathrm{e}}_{\mu }-2\delta _{\mu }^{\rho
}\overline{\mathrm{e}}_{\nu }+\overline{\mathrm{e}}^{\rho }\mathrm{e}_{\mu }%
\overline{\mathrm{e}}_{\nu },
\end{eqnarray*}%
from which it follows that%
\begin{eqnarray*}
-4\mathrm{i}S_{\mu \nu }^{\ast }\overline{\mathrm{e}}^{\rho } &=&\left(
2\delta _{\nu }^{\rho }\overline{\mathrm{e}}_{\mu }-2\delta _{\mu }^{\rho }%
\overline{\mathrm{e}}_{\nu }+\overline{\mathrm{e}}^{\rho }\mathrm{e}_{\mu }%
\overline{\mathrm{e}}_{\nu }\right) -\left( 2\delta _{\mu }^{\rho }\overline{%
\mathrm{e}}_{\nu }-2\delta _{\nu }^{\rho }\overline{\mathrm{e}}_{\mu }+%
\overline{\mathrm{e}}^{\rho }\mathrm{e}_{\nu }\overline{\mathrm{e}}_{\mu
}\right) \\
&=&4\delta _{\nu }^{\rho }\overline{\mathrm{e}}_{\mu }-4\delta _{\mu }^{\rho
}\overline{\mathrm{e}}_{\nu }+\overline{\mathrm{e}}^{\rho }\left( \mathrm{e}%
_{\mu }\overline{\mathrm{e}}_{\nu }-\mathrm{e}_{\nu }\overline{\mathrm{e}}%
_{\mu }\right) \\
&=&4\delta _{\nu }^{\rho }\overline{\mathrm{e}}_{\mu }-4\delta _{\mu }^{\rho
}\overline{\mathrm{e}}_{\nu }+4\mathrm{i}\overline{\mathrm{e}}^{\rho }S_{\mu
\nu },
\end{eqnarray*}%
from which it follows that%
\begin{eqnarray*}
S_{\mu \nu }^{\ast }\overline{\mathrm{e}}^{\rho }+\overline{\mathrm{e}}%
^{\rho }S_{\mu \nu } &=&-\mathrm{i}\left( \delta _{\mu }^{\rho }\overline{%
\mathrm{e}}_{\nu }-\delta _{\nu }^{\rho }\overline{\mathrm{e}}_{\mu }\right)
\\
&=&-\mathrm{i}\left( \delta _{\mu }^{\rho }\eta _{\nu \sigma }-\delta _{\nu
}^{\rho }\eta _{\mu \sigma }\right) \overline{\mathrm{e}}^{\sigma } \\
&=&\left( V_{\mu \nu }\right) ^{\rho }{}_{\sigma }\overline{\mathrm{e}}%
^{\sigma }.
\end{eqnarray*}
\end{proof}

\section{Auxiliary material, II: Gauge transformations}

\begin{lemma}
\label{Lemma:ScalDerU}Let $U\equiv \exp \left( u\right) $, where $u\in 
\mathbb{C}\otimes \mathbb{A}$. Then,%
\begin{equation*}
\left\langle 1,\left( \partial _{\rho }U\right) U^{-1}\right\rangle
=\left\langle 1,\partial _{\rho }u\right\rangle .
\end{equation*}
\end{lemma}

\begin{proof}
By direct calculation:%
\begin{eqnarray*}
\left\langle 1,\left( \partial _{\rho }U\right) U^{-1}\right\rangle
&=&\sum_{m=0}^{\infty }\sum_{n=0}^{\infty }\frac{1}{m!}\frac{\left(
-1\right) ^{n}}{n!}\left\langle 1,\left( \partial _{\rho }u^{m}\right)
u^{n}\right\rangle \\
&=&\sum_{m=1}^{\infty }\sum_{n=0}^{\infty }\sum_{l=0}^{m-1}\frac{1}{m!}\frac{%
\left( -1\right) ^{n}}{n!}\left\langle 1,\left[ u^{l}\left( \partial _{\rho
}u\right) u^{m-l-1}\right] u^{n}\right\rangle \\
&=&\sum_{m=1}^{\infty }\sum_{n=0}^{\infty }\sum_{l=0}^{m-1}\frac{1}{m!}\frac{%
\left( -1\right) ^{n}}{n!}\left\langle 1,u^{l}\left( \partial _{\rho
}u\right) u^{m-l-1+n}\right\rangle \\
&=&\sum_{m=1}^{\infty }\sum_{n=0}^{\infty }\sum_{l=0}^{m-1}\frac{1}{m!}\frac{%
\left( -1\right) ^{n}}{n!}\left\langle \overline{u}^{l}\overline{u}%
^{m-l-1+n},\partial _{\rho }u\right\rangle \\
&=&\sum_{m=1}^{\infty }\sum_{n=0}^{\infty }\frac{m}{m!}\frac{\left(
-1\right) ^{n}}{n!}\left\langle \overline{u}^{m-1}\overline{u}^{n},\partial
_{\rho }u\right\rangle \\
&=&\sum_{m=0}^{\infty }\sum_{n=0}^{\infty }\frac{1}{m!}\frac{\left(
-1\right) ^{n}}{n!}\left\langle \overline{u}^{m}\overline{u}^{n},\partial
_{\rho }u\right\rangle \\
&=&\left\langle \overline{U}\left( \overline{U}\right) ^{-1},\partial _{\rho
}u\right\rangle \\
&=&\left\langle 1,\partial _{\rho }u\right\rangle ,
\end{eqnarray*}%
using Eqs. (\ref{Eq:ipMoveRL}) and (\ref{Eq:ipMoveRR}), and associativity
(all quantities belonging to $\mathbb{C}\otimes \mathbb{A}$).
\end{proof}

\begin{lemma}
\label{Lemma:ScalWW}Let $W_{\rho }^{\prime }$ be given by Eq. (\ref%
{Eq:TransLorentzW}). Then,%
\begin{equation*}
\mathrm{Scal}\left( W_{\rho }^{\prime }-W_{\rho }\right) =-\mathrm{Scal}%
\left( \partial _{\rho }u\right) .
\end{equation*}
\end{lemma}

\begin{proof}
By direct calculation:%
\begin{eqnarray*}
\mathrm{Scal}\left( W_{\rho }^{\prime }\right) &=&\left\langle 1,UW_{\rho
}U^{-1}-\left( \partial _{\rho }U\right) U^{-1}\right\rangle \\
&=&\left\langle 1,UW_{\rho }U^{-1}\right\rangle -\left\langle 1,\left(
\partial _{\rho }U\right) U^{-1}\right\rangle \\
&=&\left\langle \overline{U}\left( \overline{U}\right) ^{-1},W_{\rho
}\right\rangle -\left\langle 1,\left( \partial _{\rho }U\right)
U^{-1}\right\rangle \\
&=&\left\langle 1,W_{\rho }\right\rangle -\left\langle 1,\partial _{\rho
}u\right\rangle \\
&=&\left\langle 1,W_{\rho }-\partial _{\rho }u\right\rangle \\
&=&\mathrm{Scal}\left( W_{\rho }-\partial _{\rho }u\right) ,
\end{eqnarray*}%
using Eqs. (\ref{Eq:ipMoveRL}) and (\ref{Eq:ipMoveRR}), and Lemma \ref%
{Lemma:ScalDerU}.
\end{proof}

\section{Auxiliary material, III: Octonionic identities}

The complexified octonions belong to a class of algebras called composition
algebras.

\subsection{General identities}

\begin{lemma}[Ref. \protect\cite{Springer and Veldkamp}]
The following identities hold for any composition algebra:%
\begin{eqnarray}
\left\langle xy,z\right\rangle  &=&\left\langle y,\overline{x}z\right\rangle
,  \label{Eq:ipMoveLL} \\
\left\langle xy,z\right\rangle  &=&\left\langle x,z\overline{y}\right\rangle
,  \label{Eq:ipMoveLR} \\
\left\langle z,xy\right\rangle  &=&\left\langle \overline{x}z,y\right\rangle
,  \label{Eq:ipMoveRL} \\
\left\langle z,xy\right\rangle  &=&\left\langle z\overline{y},x\right\rangle
.  \label{Eq:ipMoveRR}
\end{eqnarray}%
Even though these four expressions are pairwise related by way of the
identity $\left\langle x,y\right\rangle \equiv \left\langle y,x\right\rangle 
$, they have all been listed for presentational clarity.
\end{lemma}

\begin{lemma}[{\protect\cite[Lemma 1.2]{Zvengrowski}}]
For any $x,y,z\in \mathbb{C}\otimes \mathbb{O}$, the following identity
holds:%
\begin{equation}
x\left( \overline{y}z\right) +y\left( \overline{x}z\right) =2\left\langle
x,y\right\rangle z.  \label{Eq:Zvengrowski}
\end{equation}
\end{lemma}

\subsection{$A-B$ identities\label{SubSec:ABIdentities}}

The Cayley-Dickson construction \cite[Sec. 2.2]{Baez}, makes it clear that $%
\mathbb{O}$ is a $\mathbb{Z}_{2}$-grading over $\mathbb{H}$ (just as $%
\mathbb{H}$ is a $\mathbb{Z}_{2}$-grading over $\mathbb{C}$ which, in turn,
is a $\mathbb{Z}_{2}$-grading over $\mathbb{R}$).

\begin{lemma}[The $A-B$ Lemma]
\label{Lemma:ABLemma}Let $\mathbb{A}\cong \mathbb{H}$ and $\mathbb{B}\equiv 
\mathbb{O}\setminus \mathbb{A}$ in conjunction be a $\mathbb{Z}_{2}$-grading
of $\mathbb{O}$. Then,%
\begin{eqnarray*}
\mathbb{A}\cdot \mathbb{A} &=&\mathbb{A}, \\
\mathbb{A}\cdot \mathbb{B} &=&\mathbb{B}, \\
\mathbb{B}\cdot \mathbb{A} &=&\mathbb{B}, \\
\mathbb{B}\cdot \mathbb{B} &=&\mathbb{A},
\end{eqnarray*}%
where $\mathbb{X}\cdot \mathbb{Y}\equiv \left\{ xy\left\vert x\in \mathbb{X}%
,y\in \mathbb{Y}\right. \right\} $. By complexification, the very same holds
true, of course, for $\mathbb{C}\otimes \mathbb{A}\cong \mathbb{C}\otimes 
\mathbb{H}$ and $\mathbb{C}\otimes \mathbb{B}\equiv \left( \mathbb{C}\otimes 
\mathbb{O}\right) \setminus \left( \mathbb{C}\otimes \mathbb{A}\right) $.
\end{lemma}

\begin{lemma}
For any $a\in \mathbb{C}\otimes \mathbb{A}$ and any $b\in \mathbb{C}\otimes 
\mathbb{B}$, the following identity holds:%
\begin{equation}
ab=b\overline{a}.  \label{Eq:abIdentity}
\end{equation}
\end{lemma}

\begin{proof}
Follows immediately from $0=2\left\langle a,b\right\rangle =a\overline{b}+b%
\overline{a}=-ab+b\overline{a}$.
\end{proof}

\begin{lemma}
For any $a,a^{\prime }\in \mathbb{C}\otimes \mathbb{A}$ and any $b\in 
\mathbb{C}\otimes \mathbb{B}$, the following identities hold (note the
reversal of the order of $a$ and $a^{\prime }$):%
\begin{eqnarray}
\left( aa^{\prime }\right) b &=&a^{\prime }\left( ab\right) ,
\label{Eq:aabIdentity} \\
b\left( a^{\prime }a\right) &=&\left( ba\right) a^{\prime }.
\label{Eq:baaIdentity}
\end{eqnarray}
\end{lemma}

\begin{proof}
The two identities are related by octonionic conjugation and subsequent
renaming $\left( \overline{a},\overline{a}^{\prime },\overline{b}\right)
\rightarrow \left( a,a^{\prime },b\right) $; therefore only the first one
will be proved. By taking $\left( x,y,z\right) =\left( a^{\prime },b,%
\overline{a}\right) $ in Eq. (\ref{Eq:Zvengrowski}), and using $\overline{b}%
=-b$ and Eq. (\ref{Eq:abIdentity}), it follows that%
\begin{equation*}
0=2\left\langle a^{\prime },b\right\rangle \overline{a}=a^{\prime }\left( 
\overline{b}\overline{a}\right) +b\left( \overline{a}^{\prime }\overline{a}%
\right) =-a^{\prime }\left( ab\right) +\left( aa^{\prime }\right) b.
\end{equation*}
\end{proof}

\begin{lemma}
For any $a\in \mathbb{C}\otimes \mathbb{A}$ and any $b,b^{\prime }\in 
\mathbb{C}\otimes \mathbb{B}$, the following identities hold (note the
non-reversal of the order of $b$ and $b^{\prime }$):%
\begin{eqnarray}
\left( bb^{\prime }\right) a &=&\left( ab\right) b^{\prime },
\label{Eq:bbaIdentity} \\
a\left( b^{\prime }b\right) &=&b^{\prime }\left( ba\right) .
\label{Eq:abbIdentity}
\end{eqnarray}
\end{lemma}

\begin{proof}
The two identities are related by octonionic conjugation and subsequent
renaming $\left( \overline{a},\overline{b},\overline{b}^{\prime }\right)
\rightarrow \left( a,b,b^{\prime }\right) $; therefore only the first one
will be proved. By taking $\left( x,y,z\right) =\left( a,b,b^{\prime
}\right) $ in Eq. (\ref{Eq:Zvengrowski}), and using $\overline{b}=-b$ and
Eq. (\ref{Eq:abIdentity}), it follows that%
\begin{equation*}
0=2\left\langle a,b\right\rangle b^{\prime }=a\left( \overline{b}b^{\prime
}\right) +b\left( \overline{a}b^{\prime }\right) =-a\left( bb^{\prime
}\right) +b\left( b^{\prime }a\right) .
\end{equation*}
\end{proof}

\begin{lemma}
For any $a,a^{\prime }\in \mathbb{C}\otimes \mathbb{A}$ and any $b,b^{\prime
}\in \mathbb{C}\otimes \mathbb{B}$, the following identity holds (note the
reversal of the order of $a$ and $a^{\prime }$):%
\begin{equation}
\left( ab\right) \left( b^{\prime }a^{\prime }\right) =a^{\prime }\left(
bb^{\prime }\right) a.  \label{Eq:abbaIdentity}
\end{equation}
\end{lemma}

\begin{proof}
By direct calculation using Eqs. (\ref{Eq:bbaIdentity})-(\ref{Eq:abbIdentity}%
):%
\begin{eqnarray*}
\left( ab\right) \left( b^{\prime }a^{\prime }\right) &=&\left[ b\left(
b^{\prime }a^{\prime }\right) \right] a \\
&=&\left[ a^{\prime }\left( bb^{\prime }\right) \right] a \\
&=&a^{\prime }\left( bb^{\prime }\right) a,
\end{eqnarray*}%
where the last equality holds because $bb^{\prime }\in \mathbb{C}\otimes 
\mathbb{A}$, due to Lemma \ref{Lemma:ABLemma}, so that the relevant
quantities associate, $\left[ a,a^{\prime },bb^{\prime }\right] =0$.
\end{proof}


\begin{thebibliography}{9}
\bibitem{Okubo} Okubo, S., \textit{Introduction to Octonion and Other
Non-Associative Algebras in Physics}, Montroll Memorial Lecture Series in
Mathematical Physics, 2 (Cambridge University Press, Cambridge, 1995).

\bibitem{Springer and Veldkamp} Springer, T. A. and Veldkamp, F. D., \textit{%
Octonions, Jordan Algebras and Exceptional Groups} (Springer, Berlin, 2000).

\bibitem{Dundarer and Gursey} D\"{u}ndarer, A. R. and G\"{u}rsey, F.,
Octonionic representations of $\mathrm{SO}\left( 8\right) $ and its
subgroups and cosets, J. Math. Phys. \textbf{32}(5), 1176 (1991).

\bibitem{Baez} Baez, J. C., The Octonions, Bull. Amer. Math. Soc. \textbf{39}%
(2), 145 (2002).

\bibitem{Peskin and Schroeder} Peskin, M. E. and Schroeder, D. V., \textit{%
An Introduction to Quantum Field Theory} (Westview Press, 1995).

\bibitem{Weinberg} \label{Bib:Weinberg}Weinberg, S., \textit{The Quantum
Theory of Fields}, Vol. 1-3 (Cambridge University Press, Cambridge, 2002).

\bibitem{Zvengrowski} Zvengrowski, P., A 3-Fold Vector Product in $R^{8}$,
Comment. Math. Helv. \textbf{40}, 149 (1965/66).
\end{thebibliography}
\end{document}